\documentclass{article}
\usepackage[nohyperref, accepted]{icml2012}

\usepackage{color}
\usepackage{url}
\usepackage{verbatim} 
\usepackage{subfigure} 
\usepackage{multirow}
\usepackage{algorithm}
\usepackage{algorithmic}
\usepackage{xspace}
\usepackage{graphicx}
\usepackage{epsfig}
\usepackage{times}
\usepackage{amsmath}
\usepackage{amssymb}
\usepackage{amsthm}
\usepackage{natbib}
\usepackage{hyperref}

\DeclareMathOperator*{\argmax}{argmax}
\newcommand{\nedge}{k}

\newcommand{\hide}[1]{}
\newcommand{\xhdr}[1]{\vspace{1.7mm}\noindent{{\bf #1.}}}

\newtheorem{theorem}{Theorem}

\newcommand{\netrate}{{\textsc{Net\-Rate}}\xspace}

\newcommand{\maxinf}{{\textsc{Influ\-Max}}\xspace}

\newcommand{\eg}{\emph{e.g.}}
\newcommand{\ie}{\emph{i.e.}}

\begin{document}

\icmltitlerunning{Influence Maximization in Continuous Time Diffusion Networks}

\twocolumn[
\icmltitle{Influence Maximization in Continuous Time Diffusion Networks}

\icmlauthor{Manuel Gomez-Rodriguez$^{1,2}$}{manuelgr@stanford.edu}

\icmlauthor{Bernhard Sch\"{o}lkopf$^1$}{bs@tuebingen.mpg.de}

\icmladdress{$^1$MPI for Intelligent Systems and $^2$Stanford University}

\icmlkeywords{influence maximization, diffusion networks, social networks, temporal dynamics}

\vskip 0.3in
]

\begin{abstract}
The problem of finding the optimal set of source nodes in a diffusion network that maximizes the spread of 
information, influence, and diseases in a li\-mi\-ted amount of time depends dramatically on the underlying temporal 
dynamics of the network. However, this still remains largely unexplored to date.
To this end, given a network and its temporal dy\-na\-mics, we first des\-cribe how continuous time Markov chains \-allow us to analytically 
compute the average total number of nodes reached by a diffusion process star\-ting in a set of source nodes.
We then show that selecting the set of most influential source nodes in the con\-ti\-nuous time in\-fluence maximization problem is 
NP-hard and develop an effi\-cient \emph{approxi\-mation algorithm} with pro\-va\-ble near-optimal performance.
Experiments on synthetic and real diffusion networks show that our algorithm outperforms other state of the art algorithms 
by at least $\sim$$20$\% and is robust across different network topologies.

\end{abstract}

\section{Introduction}
\label{sec:intro}
In recent years, there has been an increasing effort in uncovering, understanding, and controlling diffusion and propagation processes in 
a broad range of domains: information propagation~\cite{leskovec2007cost}, social networks~\cite{kempe03maximizing}, viral 
marketing~\cite{richardson2002mining}, and epidemiology~\cite{wallinga04epidemic}.
%
Diffusion networks have raised many research pro\-blems, ranging from network inference~\cite{manuel10netinf, manuel11icml} to influence 
spread maximization~\cite{kempe03maximizing}. In this article, we pay attention to the \-latter problem, and we propose a method for 
continuous time influence maximization that accounts for the temporal dynamics of diffusion networks.

Influence spread maximization tackles the problem of selecting the most influential source node set of a given size in a diffusion network. A diffusion
process that starts in such an influential set of nodes is expected to reach the greatest number of nodes in the network. In information propagation, the 
problem reduces to choosing the set of blogs and news media sites that together are expected to spread a piece of news to the greatest number 
of sites. In viral marketing, it consists of identifying the most in\-fluential set of \emph{trendsetters} that together may influence the \-greatest number of
customers. Finally, in epidemiology, the in\-fluence maximization problem reduces to finding the set of individuals that together are most likely to spread
an illness or virus to the greatest percentage of the population. In this latter case, the solution of the influence maximization problem helps towards developing 
vaccination and quarantine policies.

In our work, we build on the fully continuous time model of di\-ffusion recently introduced by~\citet{manuel11icml}. This model accounts for temporally he\-te\-ro\-geneous 
in\-te\-rac\-tions wi\-thin a diffusion network -- it allows\- information (or influence) to spread at different rates across different edges, as shown in real-world examples. 
We first describe how, given a set of source nodes, we can compute the average total number of infected nodes analytically using the work of~\citet{kulkarni1986shortest}. The 
key observation is that the infection time of a node in a network with stochastic edge lengths is the length of the stochastic shortest path from the source nodes to the node. Later, 
we show that finding the optimal influential set of source nodes in the continuous time influence maximization problem is a NP-hard problem. We then provide an \emph{approximation 
al\-go\-rithm} that finds a suboptimal set of source nodes with \emph{provable guarantees} in terms of the average total number of infected nodes.

\xhdr{Related work} \citet{richardson2002mining} were the first to study influence
maximization as an algorithmic pro\-blem, motivated by marketing applications. In their work, they proposed heuristics for choosing 
a set of influential customers with a large overall effect on a network, and methods to infer the influence of each customer were developed.
\citet{kempe03maximizing} posed influence maximization in a social network as a discrete optimization 
pro\-blem. They showed that the optimal solution is NP-hard for se\-ve\-ral models of influence, and obtained the first provable approximation 
guarantees for efficient algorithms based on a natural diminishing property of the problem, submodularity.
Since then there have been substantial de\-ve\-lop\-ments that build on their seminal work. Efficient influence ma\-xi\-mi\-zation that uses heuristics to speed up 
the optimization problem has been proposed~\cite{chen2009efficient, chen2010scalable} and influence maximization has been studied on the context of competing 
cascades~\cite{bharathi2007competitive} or under additional constraints~\cite{goyal2010approximation}.

However, to the best of our knowledge, previous work on influence maximization has ignored the underlying temporal dynamics governing 
diffusion networks -- once a transmission occurs, it always occurs at the same rate or temporal scale. In contrast, we consider heterogeneous 
pairwise trans\-mi\-ssion rates, found in many real-world examples. In information propagation, news media sites and pro\-fe\-ssional bloggers 
typically report news faster than people that maintain personal blogs. In epidemiology, people meet each other with different frequencies and then 
the pairwise transmission rates between individuals within a population differ. Finally, in viral marketing, some customers make up their minds 
about a product or service quicker than {others}, and then pass recommendations on at different rates.

The main contribution of our work is twofold. First, it considers a novel continuous time formulation of the influence maximization problem in which 
information or influence can spread at different rates across different edges, as in real-world examples. Second, this continuous time approach allows 
us to analytically compute and efficiently optimize the influence (\ie, average total number of infections), avoiding the use of heuristics~\cite{chen2010scalable, chen2009efficient} 
or Monte Carlo simulations~\cite{kempe03maximizing}.


\section{Problem formulation}
\label{sec:formulation}
In this section, we build on the fully continuous time model of di\-ffu\-sion recently proposed by~\citet{manuel11icml}. We start by describing how 
the diffusion model accounts for pairwise interactions and then continue discussing some basic assumptions about di\-ffu\-sion processes. We conclude 
with a statement of the continuous time in\-fluence maximization problem.

\xhdr{Pairwise transmission likelihood}
In a diffusion network, we first need to model the pairwise in\-te\-rac\-tions between nodes. We pay attention to the general case in
which di\-ffe\-rent pairwise interactions between nodes in the network occur at different rates. Define $f(t_j | t_i ; \alpha_{i,j})$ as the 
conditional likelihood of transmission between a node $i$ and a node $j$, where $t_i$ and $t_j$ are infection times and $\alpha_{i,j}$
is the transmission rate. We assume that the likelihood depends on the pairwise transmission rate $\alpha_{i,j}$ and the time di\-fference 
$(t_j - t_i)$ (\ie, it is time shift invariant). Moreover, a node cannot be infected by a node infected later in time (\ie, $t_j > t_i$) 
and as $\alpha_{i,j} \rightarrow 0$, the expected trans\-mission time becomes arbitrarily long. 
\begin{figure*}[!t]
\centering
  \subfigure[$t_1$: $|I| = 2$, $|U_n| = 2$, $|X_n| = 4$]{\makebox[0.28\textwidth][c]{\includegraphics[width=0.22\textwidth]{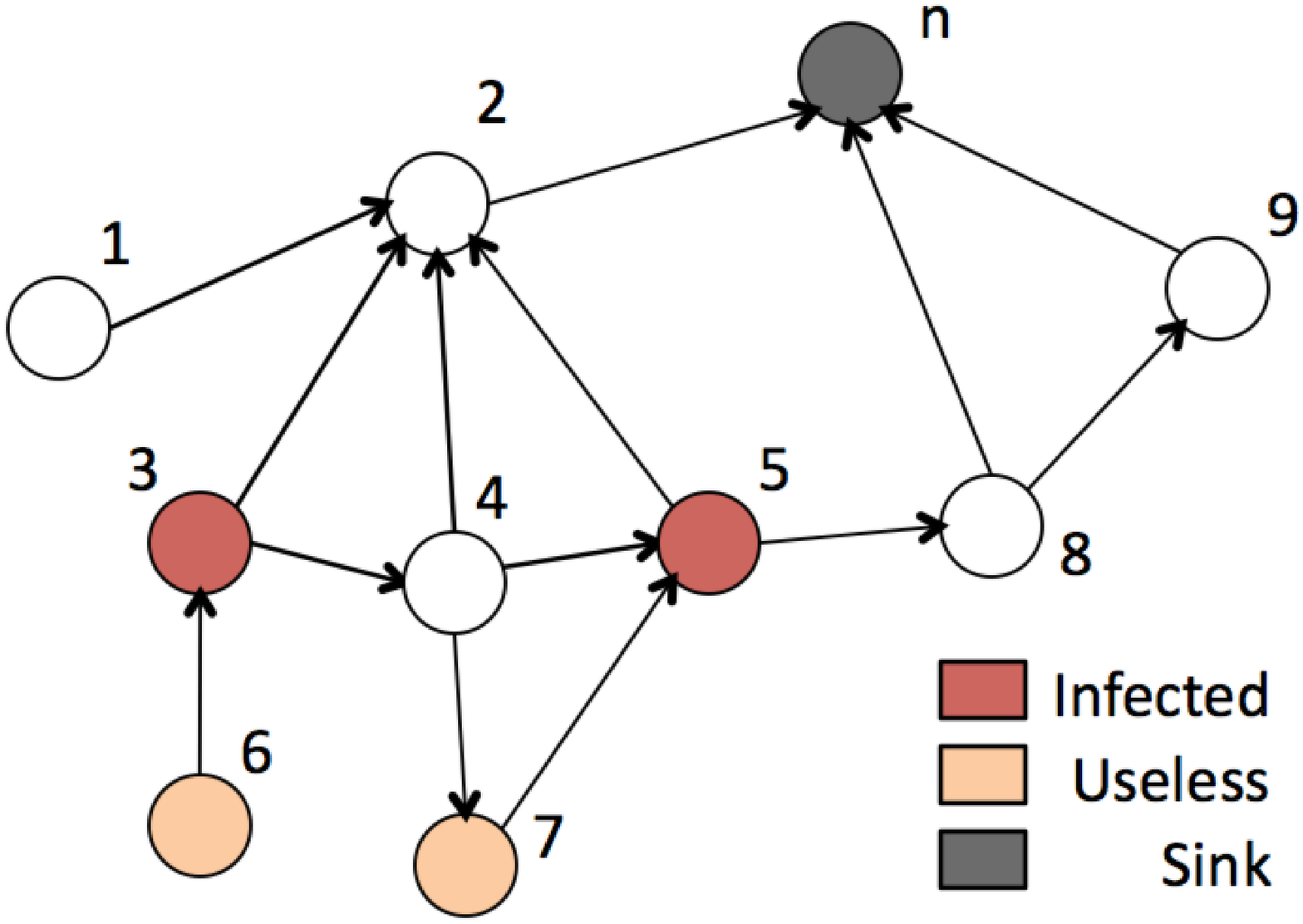} \label{fig:explanation-t0}}} \hspace{12mm}
  \subfigure[$t_2$: $|I| = 3$, $|U_n| = 4$, $|X_n| = 7$]{\makebox[0.28\textwidth][c]{\includegraphics[width=0.22\textwidth]{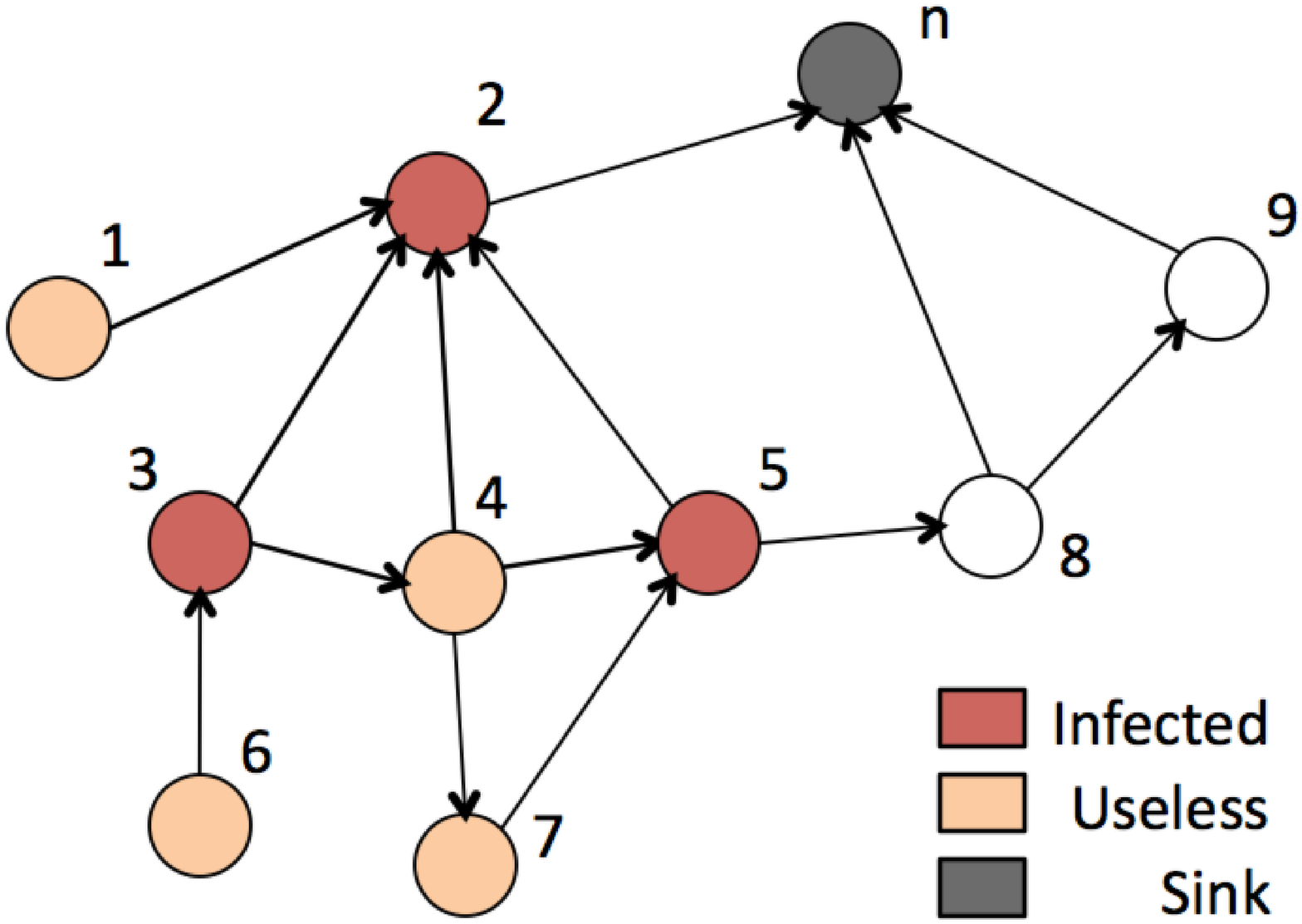} \label{fig:explanation-t1}}} \hspace{12mm}
  \subfigure[$X \in \Omega^*_n : A \subseteq X$]{\makebox[0.17\textwidth][c]{\includegraphics[width=0.11\textwidth]{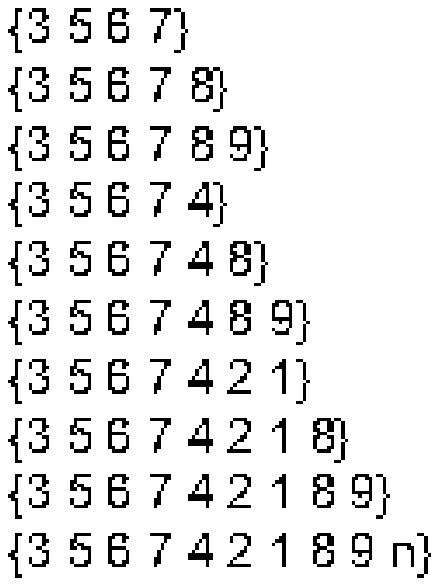} \label{fig:explanation-omega}}}
  	\caption{
	Panels (a,b): Sets of infected nodes ($I$; in red) and useless nodes ($U_n$; in orange) at two different times for a diffusion process that starts 
	in the source node set $A = \{3, 5\}$ relative to a particular sink node ($n$; in black) . Any path from a useless node to the sink node is 
	\emph{blocked} by an infected node. The set of disabled ($X_n$) nodes is simply the union of the sets of infected and useless nodes. Panel 
	(c): Sets of disabled nodes $X \in \Omega^*_n$ such that $A \subseteq X$. They represent the states that we need to describe the temporal evolution
	of a diffusion process towards the sink node $n$ that starts in the set of sources $A$.
	}
	\label{fig:explanation}
\end{figure*}

In the remainder of the paper, we consider the exponential distribution $f(t_j | t_i ; \alpha_{i,j}) \propto e^{-\alpha_{i,j} (t_j - t_i)}$ to model 
pairwise interactions for the sake of simplicity. The exponential model is a well-known parametric model for mo\-de\-ling diffusion 
and influence in social and information networks~\cite{manuel10netinf}. However, our results can easily be extended to 
diffusion networks with phase-type pairwise transmission likelihoods. This is important since the set of phase-type distributions is dense in 
the field of all positive-valued distributions and it can used to approximate power-laws, which have been also used for modeling diffusions in 
social networks~\cite{multitree12icml}, Rayleigh distributions, which have been used in epi\-de\-mio\-lo\-gy~\cite{wallinga04epidemic}, and also 
subprobability distributions, which enable us to describe two step traditional diffusion models~\cite{kempe03maximizing}, in which with 
probability $(1-\beta)$ an infection may never occur.

\xhdr{Continuous time diffusion process}
We consider diffusion and propagation processes that occur over static networks with known (or inferred) connectivity and trans\-mission rates. 
A diffusion process starts when a source node set $A$ becomes infected at time $t = 0$ by action of an external source to the network. 
Then, source nodes try to infect their children (\ie, neighbors that they can reach directly through an outgoing edge). Once a child $i$ gets infected 
at time $t_i$, it tries to infect her own children, and so on. For some pairwise transmission likelihoods, it may happen that $t_i \rightarrow \infty$ 
and child $i$ is never infected.
Here, we assume that a node $i$ becomes infected as soon as one of its parents (\ie, neighbors that are able to reach node $i$ through an outgoing edge) 
infects it, and later infections by other parents do not contribute anymore towards the evolution of the diffusion process. As a consequence of this assum\-ption, at 
any time $t \geq 0$ there may be some nodes and edges in the network that are useless for the spread of the information (be it in the form of a meme, a sales 
decision or a virus) towards a specific node $n$. If these nodes get infected and transmit the information to other nodes, this information can only reach $n$ 
through previously infected nodes. Therefore, the infection time $t_n$ of node $n$ does not depend on these nodes.


Finally, given a diffusion process that started in the set of source nodes $A$, we define $N(A ; T)$ as the number of nodes infected up to time $T$ and then define the 
influence function $\sigma(A; T)$ as the average total number of nodes infected up to time $T$, \ie, $\sigma(A; T) = \mathbb{E} N(A; T)$.

\xhdr{Continuous time influence maximization problem}
Our goal is to find the set of source nodes $A$ in a diffusion network $G$ that maximizes the influence function $\sigma(A; T)$. In other words, the set of source 
nodes $A$ such that a diffusion process in $G$ reaches, on average, the greatest number of nodes before a window cut off $T$. Thus, we aim to solve:
\begin{equation}
  A^* = \argmax_{|A|\leq\nedge} \sigma(A ; T),
  \label{eq:max-influence-opt}
\end{equation}
where the source set $A$ is the variable to optimize and the time horizon $T$ and the source set cardinality $k$ are cons\-tants.

\section{Proposed algorithm}
\label{sec:proposed}
We start this section by describing how to evaluate the influence function $\sigma(A; T)$ for any set of sources $A$ in a network $G$ using the work
of~\citet{kulkarni1986shortest}. The key observation is that the infection time of a node in a network with stochastic edge lengths is the length of the 
stochastic shortest path from the source nodes to the node. Then, we show that the continuous time in\-fluence ma\-xi\-mi\-zation problem defined by 
Eq.~\ref{eq:max-influence-opt} is NP-hard. Finally, we show how to efficiently find a \emph{provable near-optimal} solution to our maximization problem
by exploiting a natural diminishing returns property of our objective function.

\xhdr{Evaluating the influence}
The influence function depends on the probability of infection of every node in the network as follows:
\begin{align} \label{eq:sigma}
\sigma(A; T) &= \mathbb{E} N(A; T) = \sum_{n=1}^{N} P(t_n \leq T | A),
\end{align}
where $t_n$ is the infection time of node $n$, $A$ is the set of source nodes, and $T$ is the time horizon or time window cut-off. Therefore, we need 
to compute the probability of infection $P(t_n \leq T | A)$ for each node $n$ in the network. Note that whenever $n \in A$, the probability of infection 
$P(t_n \leq T | A)$ is trivially $1$. We will refer to node $n$ as sink node.

Revisiting the basic assumptions about a diffusion process that we presented in Section~\ref{sec:formulation}, we recall some definitions to
describe its temporal evolution as in~\citet{kulkarni1986shortest}. Given a diffusion network $G = (V, E)$, a set of nodes $B \subset V$, and a 
node $n \in V$, we define the set of nodes blocked by or dominated by $B$:
\begin{equation*}
\begin{split}
S_n(B) =& \{u \in V:\, \mbox{any path from $u$ to $n$ in $G$ visits}\\ 
&\mbox{ at least one node in $B$}\}.
\end{split}
\end{equation*}
By definition, $B \subseteq S_n(B)$ and $S_n(S_n(B)) = S_n(B)$. We now define the set $\Omega^*_n$ as:
\begin{equation*}
\Omega^*_n = \{X \subset V :\, X = S_n(X)\}.
\end{equation*}
In words, all nodes in $X \in \Omega^*_n$ block only themselves re\-la\-tive to the sink node $n$. We can find all sets in $\Omega^*_n$ 
efficiently~\cite{georgiadis2004finding, provan1996paradigm}. In particular, we are able to find each $X \in \Omega^*_n$ in 
time $O(|V|)$. However, in dense networks, $|\Omega^*_n|$ can be exponentially large and lead to a worst-case non polynomial 
time algorithm.
In order to illustrate this, we compute $\max_n |\Omega^*_n|$ across $1,000$ random source sets with $|S| = 5$ and $|S| = 10$ for several 
$256$-node hierarchical networks of increasing network density. 
We observe that $\max_n |\Omega^*_n| < 85$ for all networks up to $2$ edges per node in average. However, $\max_n |\Omega^*_n|$ grows 
quickly for higher network densities (\eg, $\max_n |\Omega^*_n| < 7750$ for a network with $2.5$ edges per node in average). In order to overcome 
this drawback, we will propose several speed-ups (LTP and LSN) that provide approximate solutions or sparsify the networks as in~\citet{bonchi11kdd}.

Given a diffusion process that starts in a set of source nodes $A$, a sink node $n$ and any time $t \geq 0$, we denote the set of infected
nodes as $I(t | A)$, the set of useless nodes as $U_n(t | A)$, and the set of \emph{disabled} nodes (\ie, infected or useless) as $X_n(t | A)$. 
Useless nodes are nodes that if they get infected and transmit the information to other nodes, this information can only reach the sink node 
$n$ through pre\-vious\-ly infected nodes.
Figures~\ref{fig:explanation-t0} and~\ref{fig:explanation-t1} illustrate the set of infected nodes ($I$) and the set of useless nodes ($U_n$) for a 
diffusion process in a small network at two different times. Note that the set of disabled nodes ($X_n$) is composed of the sets of infected ($I$) and 
useless nodes ($U_n$). 
By definition of $S_n(\cdot)$, $U_n(t | A) = S_n(I(t | A)) \backslash I(t | A)$ and $X_n(t | A) = S_n(I(t | A))$. Now, by studying 
the temporal evolution of $X_n(t | A)$ we will be able to compute $P(t_n \leq T | A)$.
\begin{table}[t]
    \small
    \centering
    \begin{tabular}{l c c c c}
       \textbf{Algorithm} & \textbf{$\bf |A|$} & \textbf{$\bf{\sigma(A ; 0.1)}$} & \textbf{$\bf{\sigma(A ; 0.5)}$} & \textbf{$\bf{\sigma(A ; 1.0)}$} \\
        \hline
        \multirow{3}{*}{Enumeration} & 1 & 3.05 & 8.95 & 13.90\\
        & 3 & 8.04 & 18.01 & 21.70\\
        & 5 & 11.60 & 22.17 & 25.59\\
        \hline
        \multirow{3}{*}{\maxinf} & 1 & 3.05 & 8.95 & 13.90\\
        & 3 & 8.04 & 18.01 & 21.70\\
        & 5 & 11.60 & 22.17 & 25.59\\
        \hline
        \multirow{3}{*}{Greedy} & 1 & 3.05 &  7.70 & 9.31\\
        & 3 & 6.18 & 12.70 & 15.88\\
        & 5 & 8.62 & 16.05 & 19.70\\
        \hline
        \multirow{3}{*}{PMIA} & 1 & 1.20 & 1.67 & 1.89 \\
        & 3 & 4.90 & 11.67 & 16.80 \\
        & 5 & 8.90 & 18.03 & 22.02\\
        \hline
        \multirow{3}{*}{SP1M} & 1 & 2.15 & 8.04 & 11.66\\
        & 3 & 4.88 & 10.75 & 13.14\\
        & 5 & 7.96 & 13.95 & 16.16\\
        \hline
        \multirow{3}{*}{Random} & 1 & 1.61 & 3.77 & 5.34\\
        & 3 & 4.63 & 9.29 & 11.84\\
        & 5 & 7.39 & 13.36 & 16.04
    \end{tabular}
    \caption{Influence $\sigma(A ; T)$ that enumeration, \maxinf and several other baselines achieve in a small Kronecker core-periphery network with 35 nodes and 39 
    edges for different time horizon values $T$ and number of sources $|A|$. \maxinf always achieves the optimal influence that exhaustive search gives but several order 
    of magnitude faster.}
    \label{tab:performance-small-network}
\end{table}

First, for a diffusion process that starts in the set of source nodes $A$, it can be shown that the set of disabled nodes $X_n(t | A)$ at any time $t \geq 0$ 
belongs to $\Omega^*_n$.
\begin{theorem} (\citet{kulkarni1986shortest})
Given a set of source nodes $A$, a sink node $n$ and any time $t \geq 0$, $X_n(t | A) \in \Omega^*_n$.
\end{theorem}
%
%
Figure~\ref{fig:explanation-omega} enumerates all sets of disabled nodes $X \in \Omega^*_n$ such that $A \subseteq X$ for the small network depicted in 
Fi\-gures~\ref{fig:explanation-t0} and~\ref{fig:explanation-t1}. They represent the states that we need to describe the evolution of a diffusion process that starts in the
set of sources $A$ relative to the sink node $n$. Now, assu\-ming independent pairwise exponential transmission likelihoods in the diffusion network, the following Th.
applies:
\begin{theorem} (\citet{kulkarni1986shortest}) \label{th:ctmc}
Given a set of source nodes $A$, a sink node $n$ and independent pairwise exponential transmission likelihoods $f(t_j | t_i ; \alpha_{i,j})$, $\{X_n(t | A), t \geq 0\}$ 
is a continuous time Markov chain (CTMC) with state space $\{X : X \in \Omega^*_n, A \subseteq X\}$ and infinitesimal generator matrix $Q = [q(D, B)]$ ($D, B \in 
\{X : X \in \Omega^*_n, A \subseteq X\}$) given by:
\begin{equation*}
q(D, B) = \left\{
	\begin{array}{l l}
    	\sum_{(i,j) \in C_v(D)} \alpha_{i,j} & \hspace{-2.5mm} \exists{v} : B = S_n(D \cup \{v\}), \\
 	-\sum_{(i,j) \in C(D)} \alpha_{i,j} & \hspace{-2.5mm} B = D, \\
	0 & \hspace{-2.5mm} \mbox{otherwise.}
	\end{array}\right.
\end{equation*}
where $C(D)$ is the unique minimal cut between $D$ and $\bar{D} = V \backslash D$ and $C_v(D) = \{(u,v) \in C(D)\}$.
\end{theorem}
\begin{figure*}[!t]
\centering
  \subfigure[Forest Fire]{\includegraphics[width=0.22\textwidth]{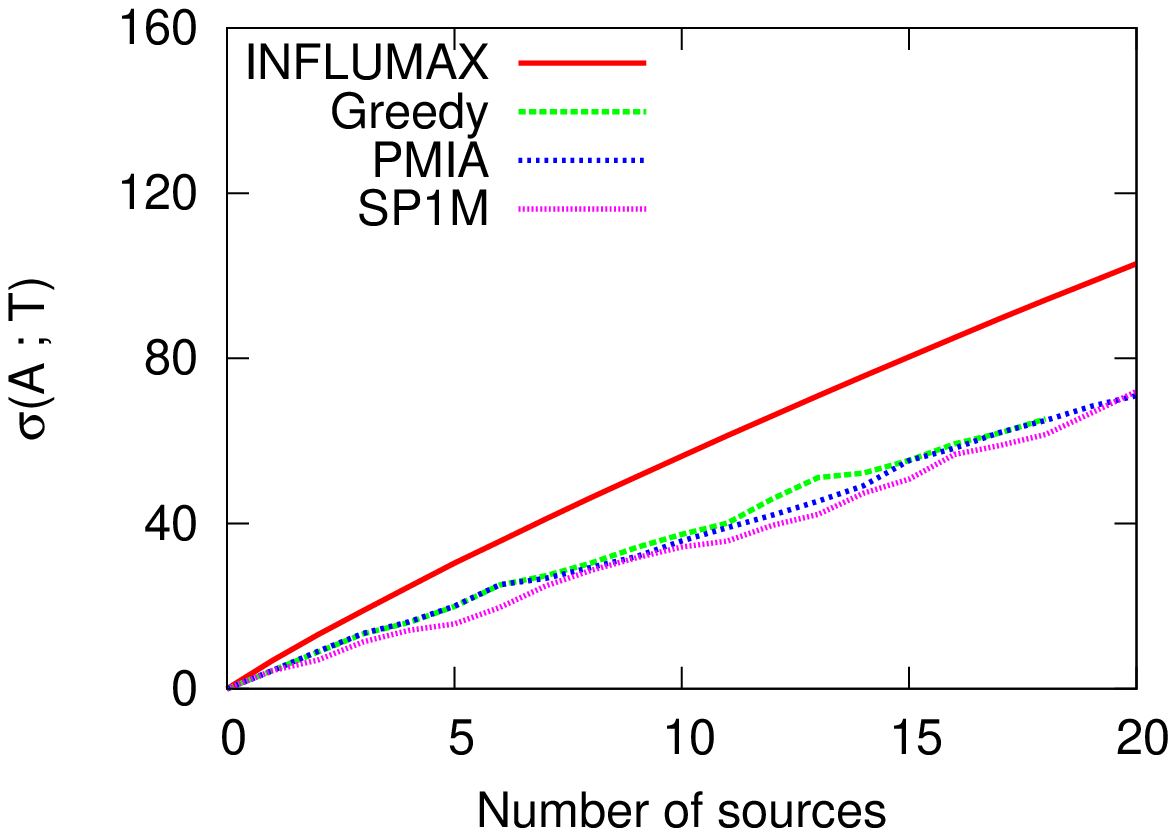} \label{fig:ff1}} \hspace{9mm}
  \subfigure[Random Kronecker]{\includegraphics[width=0.22\textwidth]{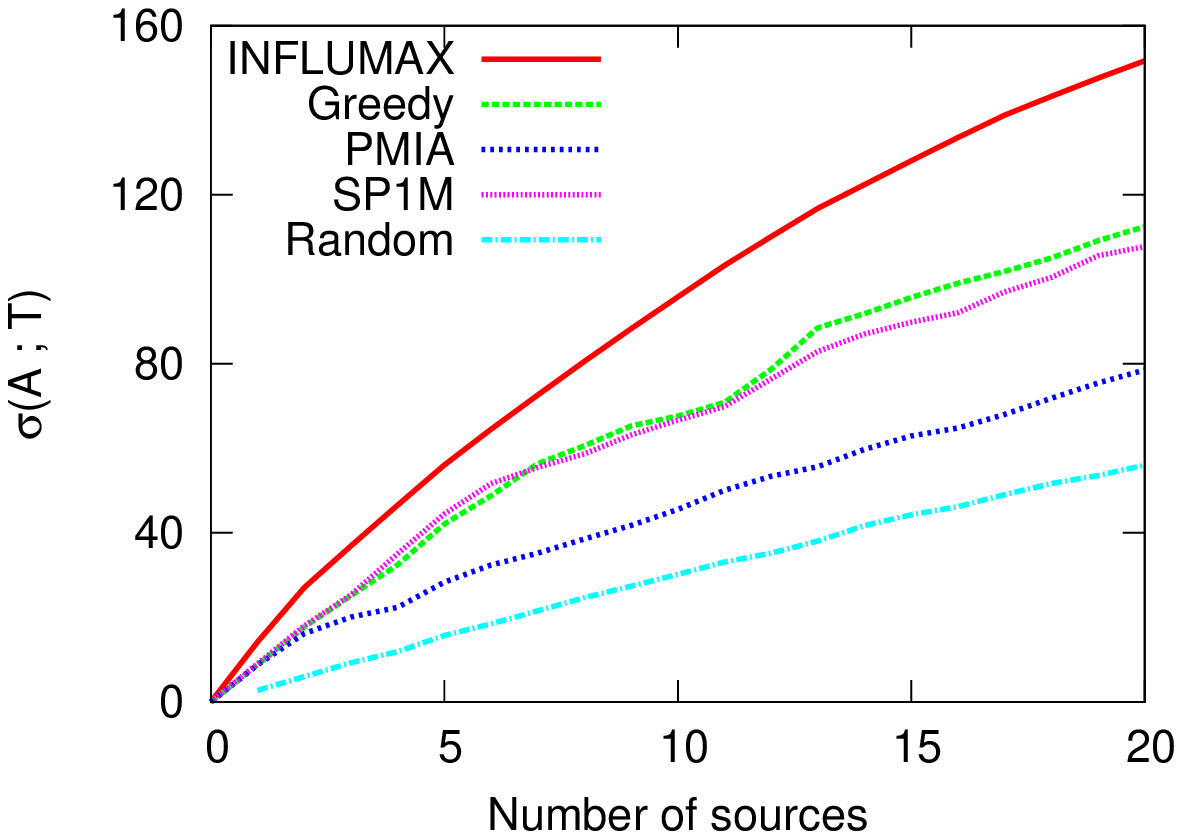} \label{fig:kro}} \hspace{9mm}
  \subfigure[Hierarchical Kronecker]{\includegraphics[width=0.22\textwidth]{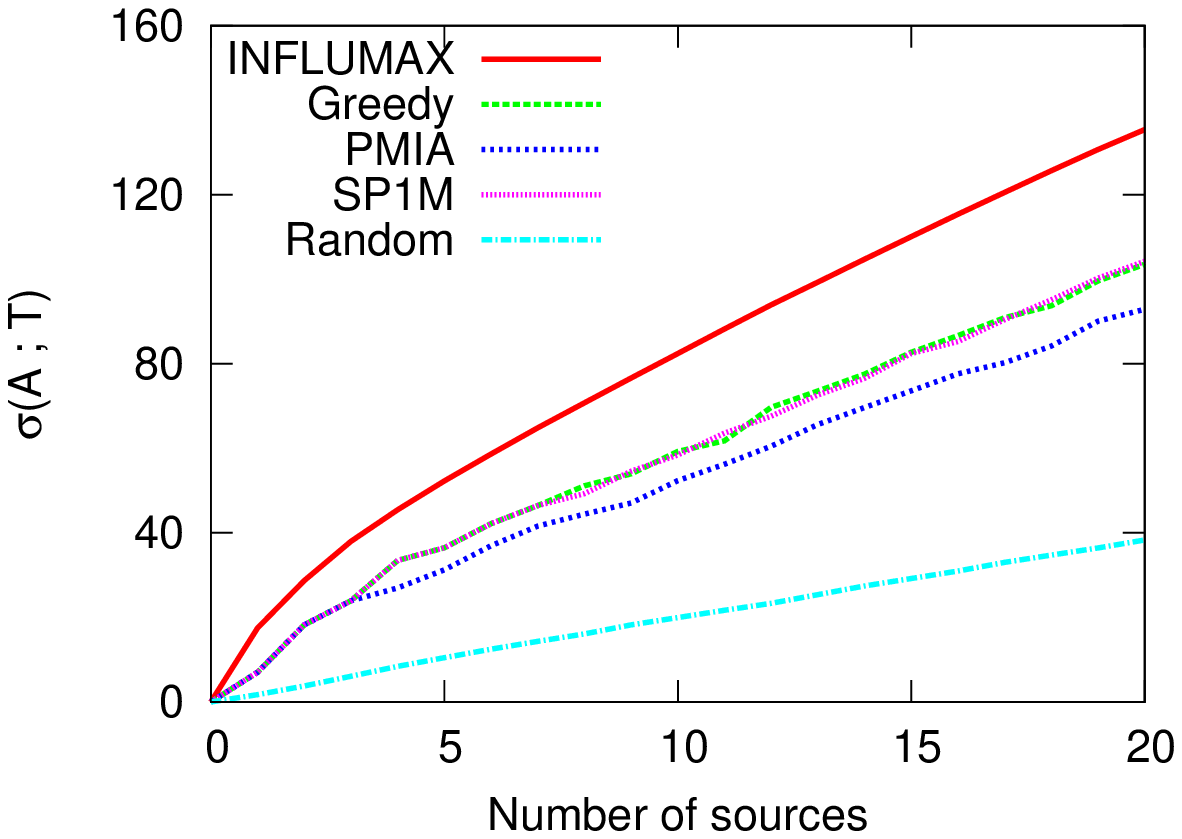} \label{fig:real}} 
	\caption{
	Panels plot influence $\sigma(A ; T)$ (\ie, average number of infected nodes) for $T=1$ and transmission rates drawn from \mbox{$\alpha \sim U(0, 5)$} against number of sources. 
	(a): 1,024 node Forest Fire network. 
	(b): 512 node random Kronecker network.
	(c): 1,024 node hierarchical Kronecker network. The proposed algorithm \maxinf outperforms all other methods typically by at least 20\%.
	}
	\label{fig:performance}
\end{figure*}

Finally, let $t_n$ be the \emph{length} of the fastest (shortest) directed path from any of the nodes in $A$ to the sink node $n$ in the directed acyclic graph (DAG) induced 
by the diffusion process on network $G$. By construction of the CTMC $\{X_n(t | A), t \geq 0\}$ in Theorem~\ref{th:ctmc},
\begin{equation*}
t_n = \min\{t \geq 0 :\, X_n(t | A) = S_N | X_n(0 | A) = S_1\},
\end{equation*}
where $S_1$ and $S_N$ denote respectively the first and last state of the CTMC. The \emph{length} of the fastest (shortest) path is thus equivalent to the time until the 
CTMC $\{X_n(t | A), t \geq 0\}$ becomes absorbed in the final state $S_N$ starting from state $S_1$ (\ie, the state in which only the source nodes in $A$ are infected). 
Then, computing the probability of infection of the sink node $P(t_n \leq T | A)$ reduces to com\-pu\-ting the distribution of time of the sink state of the CTMC. Such distributions 
are called continuous phase-type distributions. Their generator matrix $Q$ and the cumulative density function satisfy~\cite{gikhman2004theory}:
\begin{equation*}
P(t_n \leq T | A) = 1-[1 \mathbf{0}]' e^{S T} \mathbf{1},\, \mbox{where }\, Q = \begin{bmatrix}
S & S^0 \\
\mathbf{0}' & 0
\end{bmatrix},
\end{equation*}
where $e^{S T}$ denotes the exponential matrix, $S$ is the submatrix of $Q$ that results from removing the column and row associated to the last state $S_N$, and 
$S^0 = -S \mathbf{1}$. By construction, $\{X_n(t | A), t \geq 0\}$ has the structure of a DAG and it is usually sparse. Then, $S$ is upper triangular, sparse and $e^{S T}$ can 
be computed efficiently. 

As noted in~\citet{kulkarni1986shortest}, this approach can be easily extended to diffusion networks with phase-type transmission likelihoods, which can approximate power-laws, 
Rayleigh or subprobability distributions. 

\xhdr{Maximizing the influence}
We have shown how to ana\-lytically evaluate our objective function $\sigma(A; T)$ for any set of sources $A$. However, optimizing 
$\sigma(A; T)$ with res\-pect to the set of sources $A$ seems to be a cumbersome task and naive brute-force search over all $k$ node 
sets is intractable even for relatively small networks. Indeed, we cannot expect to find the optimal solution to the continuous time influence 
maximization problem defined by Eq.~\ref{eq:max-influence-opt} since it is NP-hard:
\begin{theorem}\label{thm:np_maxinfluence}
Given a network $G = (V, E)$, a set of nodes $A \subseteq V$ and a time horizon $T$, the continuous time influence maximization problem defined by 
Eq.~\ref{eq:max-influence-opt} is NP-hard.
\end{theorem}
\begin{proof} 
If we let $T \rightarrow \infty$, the independent cascade model is a particular case of our continuous time diffusion model. Then, our problem 
is NP-hard by applying Th.~2.4 in~\citet{kempe03maximizing}.
\end{proof}

By construction, $\sigma(\emptyset, T) = 0$ and $\sigma(A; T) \geq 0$. It also \-follows trivially that $\sigma(A; T)$ is monotonically non\-de\-crea\-sing in the set of source 
nodes $A$, \ie, $\sigma(A; T) \leq \sigma(A'; T)$, whenever $A \subseteq A'$.  Fortunately, we now show that the objective function $\sigma(A; T)$ is a submodular 
function in the set of source nodes $A$.  A set function $F: 2^{W}\rightarrow\mathbb{R}$ mapping subsets of a finite set $W$ to the real numbers is submodular 
if whenever $A\subseteq B\subseteq W$ and $s\in W\setminus B$, it holds that $F(A\cup\{s\})-F(A)\geq F(B\cup\{s\})-F(B)$, i.e., adding $s$ to the set $A$ provides 
a bigger marginal gain than adding $s$ to the set $B$. By this natural diminishing returns property, we are able to find a \emph{provable near-optimal} solution to 
our problem:
\begin{theorem} \label{thm:submodular_maxinfluence}
Given a network $G = (V, E)$, a set of nodes $A \subseteq V$ and a time horizon $T$, the influence function 
$\sigma(A; T)$ is a submodular function in the set of nodes $A$.
\end{theorem}
\begin{proof}
We follow the proof of Th.~2.2 in~\citet{kempe03maximizing}. 
For simplicity, we assume that the infection time of all nodes in $A$ is $t = 0$; the results generalize tri\-vially. 
Consider the probability distribution of all possible time diffe\-ren\-ces between each pair of nodes in the network. Thus, 
given a sample $\mathbf{\Delta t}$ in the probability space, we define $\sigma_{\mathbf{\Delta t}}(A; T)$ as the total number 
of nodes infected in a time less than or equal to $T$ for $\mathbf{\Delta t}$.

Define $R_{\mathbf{\Delta t}}(k; T)$ as the set of nodes that can be reached from node $k$ in a time shorter than $T$. It follows 
tri\-vially that $\sigma_{\mathbf{\Delta t}}(A; T) = |\cup_{k \in A} R_{\mathbf{\Delta t}}(k; T)|$. Define $R_{\mathbf{\Delta t}}(k | N ; T)$ 
as the set of nodes that can be reached from node $k$ in a time shorter than $T$ and at the same time cannot be reached in a time 
shorter than $T$ from any node in the set of nodes $N \subseteq V$. It follows that $|R_{\mathbf{\Delta t}}(k | N; T)| \geq |R_{\mathbf{\Delta t}}(k | N' ; T)|$ 
for the sets of nodes $N \subseteq N'$.

Consider now the sets of nodes $A \subseteq A' \subseteq V$, and a node $a$ such that $a \notin A'$. Using the definition of 
submodularity,
\begin{align*}
\sigma_{\mathbf{\Delta t}}(A \cup \{a\}; T) - \sigma_{\mathbf{\Delta t}}(A; T) &= |R_{\mathbf{\Delta t}}(a | A ; T)| \\
&\hspace{-14mm}\geq |R_{\mathbf{\Delta t}}(a | A' ; T)| \\
&\hspace{-14mm}= \sigma_{\mathbf{\Delta t}}(A' \cup \{a\}; T) - \sigma_{\mathbf{\Delta t}}(A'; T),
\end{align*}
and thus $\sigma_{\mathbf{\Delta t}}(A; T)$ is submodular. Then, it follows that $\sigma(A; T)$ is also
submodular. 
\end{proof}

A well-known approximation algorithm to maximize mo\-no\-to\-nic submodular functions is the \emph{greedy algorithm}. It adds 
nodes to the source node set $A$ sequentially. In step $k$, it adds the node $a$ which maximizes the \emph{marginal gain} 
$
\sigma(A_{k-1} \cup \{a\} ; T) - \sigma(A_{k-1} ; T).
$
The greedy algorithm finds a source node set which achieves at least a constant fraction $(1-1/e)$ of the optimal~\cite{nemhauser1978analysis}. 

Moreover, we can also use the submodularity of $\sigma(A ; T)$ to acquire a tight \emph{online} bound on the solution quality obtained by \emph{any} 
algorithm:
\begin{theorem}[\citet{leskovec2007cost}]
  For a source set $\hat{A} \subseteq V$ with $k$ sources and a node $a \in V \backslash \hat{A}$, let $\delta_{a} = \sigma(\hat{A} \cup \{a\} ; T) - \sigma(\hat{A} ; T)$ and 
  $a_1, \ldots a_{k}$ be the sequence of $k$ nodes with $\delta_{a}$ in decreasing order. Then,
  $
  \max_{|A| \leq k} \sigma(A; T) \leq \sigma(\hat{A}; T) + \sum_{i=1}^{k} \delta_{a_i}.
  $
\label{th:online-bound}
\end{theorem}
Lazy evaluation~\cite{leskovec2007cost} can be employed to speed-up the computation of the on-line bound for our algorithm, that we will refer as \maxinf.

\xhdr{Speeding-up \maxinf} We can speed up our algorithm by implementing the following speed-ups:

\emph{Lazy evaluation} (LE, \citet{leskovec2007cost}): it dramatically reduces the number of evaluations of marginal gains by exploiting the submodularity of $\sigma(A ; T)$.

\emph{Localized source nodes} (LSN): for each node $n$, we speed up the computation of $P(t_n \leq T | A)$ by ignoring any $a \in A$ whose shortest path to $n$ traverses 
more than $m$ nodes.

\emph{Limited transmission paths} (LTP): for each node $n$, we speed up the computation of $P(t_n \leq T | A)$ by ignoring any path from $a \in A$ to $n$ that traverses more than 
$m$ nodes.

LSN and LTP should be used with care since they provide an approximate $P(t_n \leq T | A)$. In the remainder of this article, if not specified, we run \maxinf with LE but 
avoid using LSN and LTP.

\section{Experimental evaluation}
\label{sec:evaluation}
We evaluate our algorithm \maxinf on (i) synthetic networks that mimic the structure of real networks and on (ii) real networks inferred from the MemeTracker 
dataset\footnote{Data available at \url{http://memetracker.org}} by using \netrate'{}s public im\-ple\-men\-tation~\cite{manuel11icml}.
We show that \maxinf outperforms three state of the art algorithms: the traditional gree\-dy al\-go\-rithm~\cite{kempe03maximizing}, 
PMIA~\cite{chen2010scalable} and SP1M~\cite{chen2009efficient}.
\begin{figure}[!t]
\centering
  \subfigure[Small network]{\includegraphics[width=0.15\textwidth]{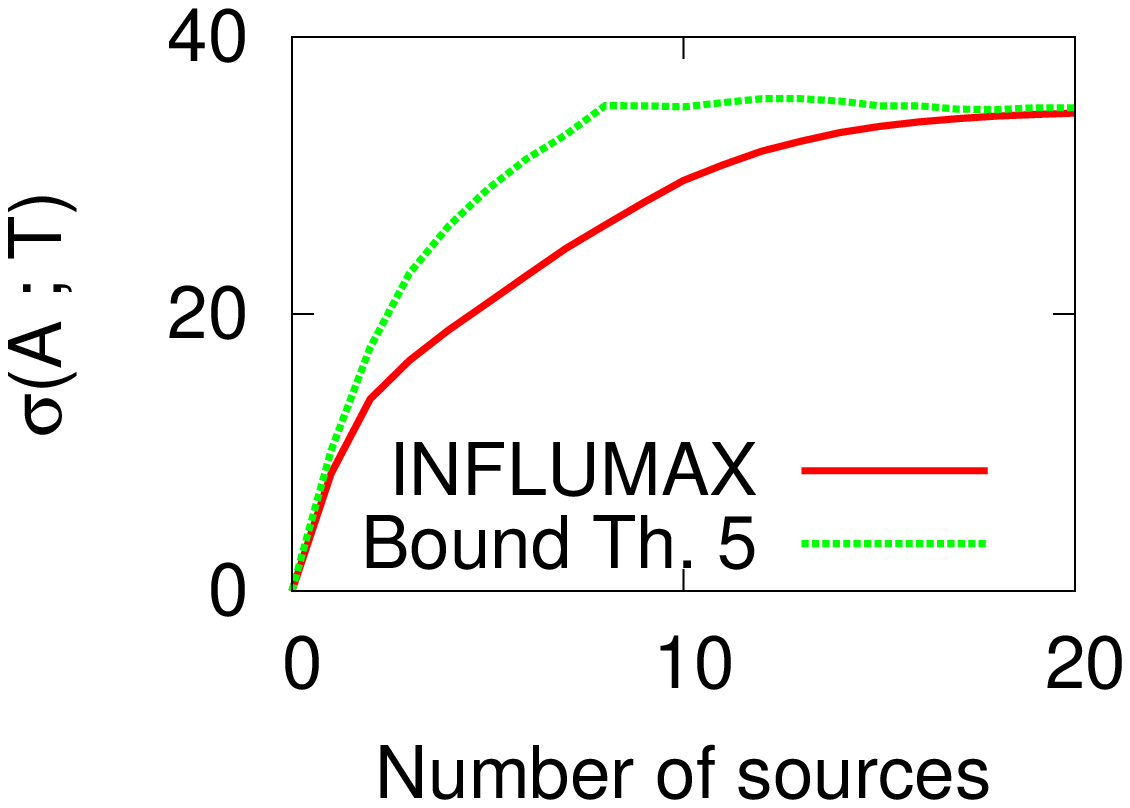} \label{fig:bound-small-kronecker}}
  \subfigure[Large network]{\includegraphics[width=0.15\textwidth]{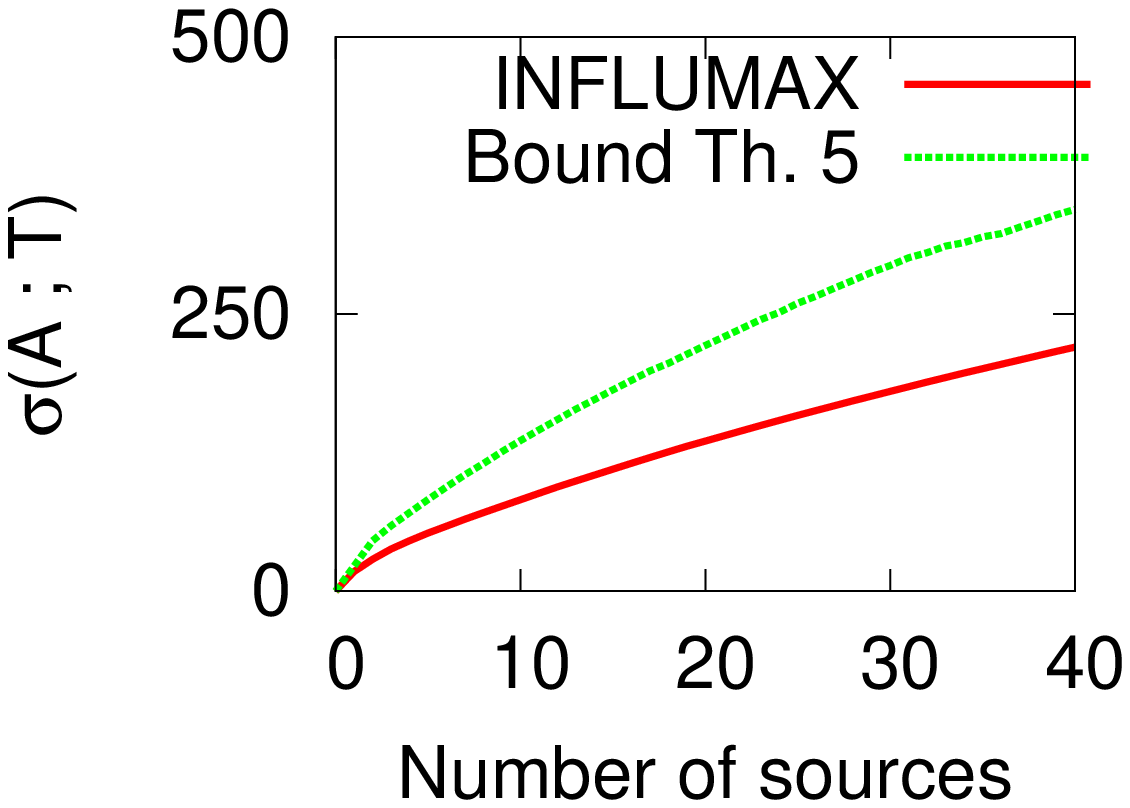} \label{fig:bound-big-kronecker}}
  \subfigure[Real network]{\includegraphics[width=0.15\textwidth]{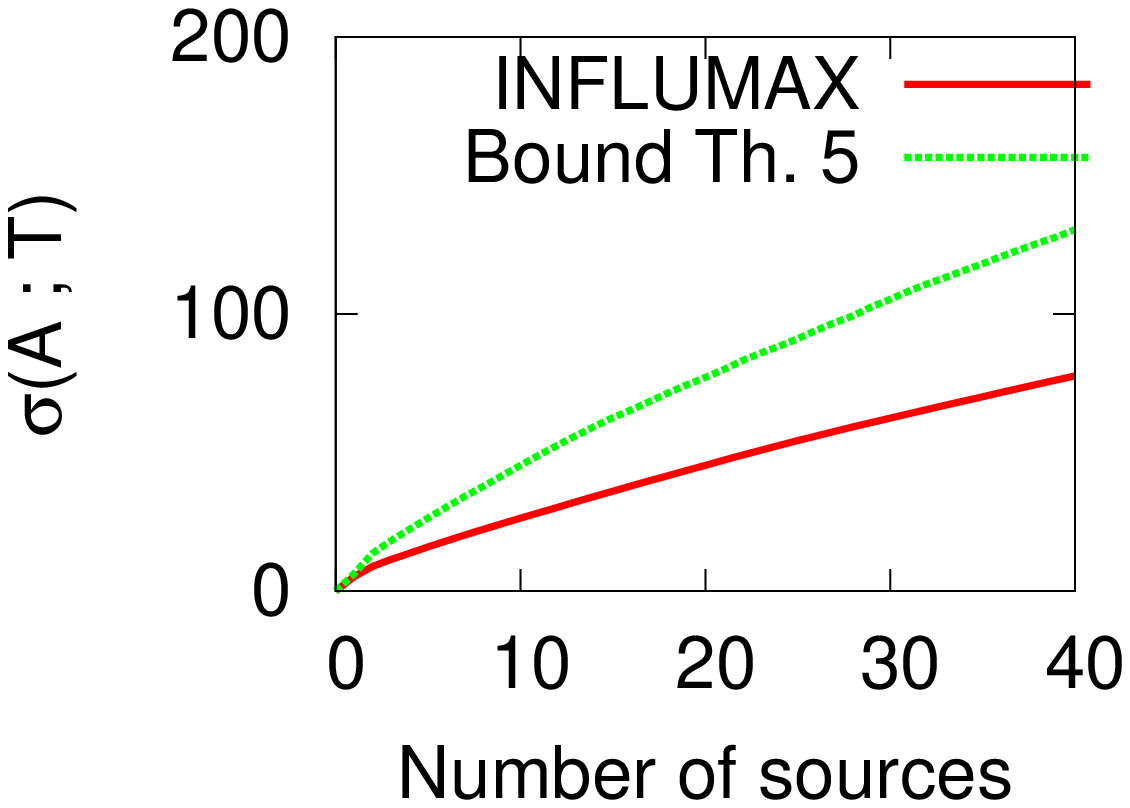} \label{fig:bound-real}}
  \caption{Influence $\sigma(A; T)$ achieved by \maxinf in com\-pa\-ri\-son with the online upper bound from Theorem~\ref{th:online-bound} for $T = 1$. (a) 35-node core-periphery Kronecker 
  network. (b) 1,024 node hierarchical Kronecker network. (c) 1,000 node real diffusion network that we infer from hyperlinks cascades ($T = 1$).}
\label{fig:online-bound}
\end{figure}

\subsection{Experiments on synthetic data}

\xhdr{Experimental setup} We perform experiments on two types of synthetic networks that mimic the structure of directed social networks: Kronecker~\cite{leskovec2010kronecker} and 
Forest Fire (scale free)~\cite{barabasi99emergence} networks. We consider three types of Kronecker networks with very different structure: random~\cite{erdos60random} (parameter 
matrix [0.5, 0.5; 0.5, 0.5]), hierarchial~\cite{clauset08hierarchical} ([0.9, 0.1; 0.1; 0.9]) and core-periphery~\cite{leskovec2010kronecker} ([0.9, 0.5; 0.5, 0.3]).

First, we generate a network $G$ using one of the network models cited above. Then, we draw a transmission rate for each edge $(j, i) \in G$ from a uniform 
distribution. We can control the transmission rate variance across edges in the network by tuning the parameters values of the distribution. In social networks, 
transmission rates model how fast information spreads across the network. Given $G$ and the transmission rates $\alpha_{j,i}$, our aim is to find the most influential 
subset of $k$ nodes, \ie, the subset of nodes that maximizes the spread of information up to a time $T$. In the traditional gree\-dy al\-go\-rithm, PMIA and SP1M, we 
ignore any of the transmission rates and consider all network edges to be active with probability $1$, \ie, we do not consider the temporal dynamics. We did not need 
to use Montecarlo in the traditional greedy algorithm since we assume all edges to be always active.

\xhdr{Solution quality} First, we compare \maxinf to exhaustive search and several state of the art algorithms on a small network. By studying a small network in 
which exhaustive search can be run, we are able estimate exactly how far \maxinf is from the NP-hard to find optimum. We then compare \maxinf to the state of the art 
on diffe\-rent large networks. Running exhaustive search on large networks is computationally too expensive and we compute instead the tight on-line bound from 
Th.~\ref{th:online-bound}.

We compare \maxinf to several state of the art methods on a small core-periphery Kronecker network with 35 nodes and 39 edges and transmission rates drawn from a 
uniform distribution $\alpha \sim U(0, 10)$. 
%
We summarize the results in Table~\ref{tab:performance-small-network}. In addition to \maxinf and three state of the art methods, we also run a baseline that simply chooses the set of 
sources randomly. For all methods, we compute the influence they achieve by evaluating Eq.~\ref{eq:sigma} for the set of sources selected by them.
Surprisingly, \maxinf achieves in most cases the optimal influence that exhaustive search gives but several order of magnitude faster. In other words, the solution given by \maxinf 
may be in practice much closer to the NP-hard to find optimum than $(1-1/e)$, the theoretical guarantee given by~\citet{nemhauser1978analysis}, and it outperforms 
other methods by 20\%.
\begin{figure}[t]
\centering
 \subfigure[$\sigma(A; T)$ vs. T]{\includegraphics[width=0.2\textwidth]{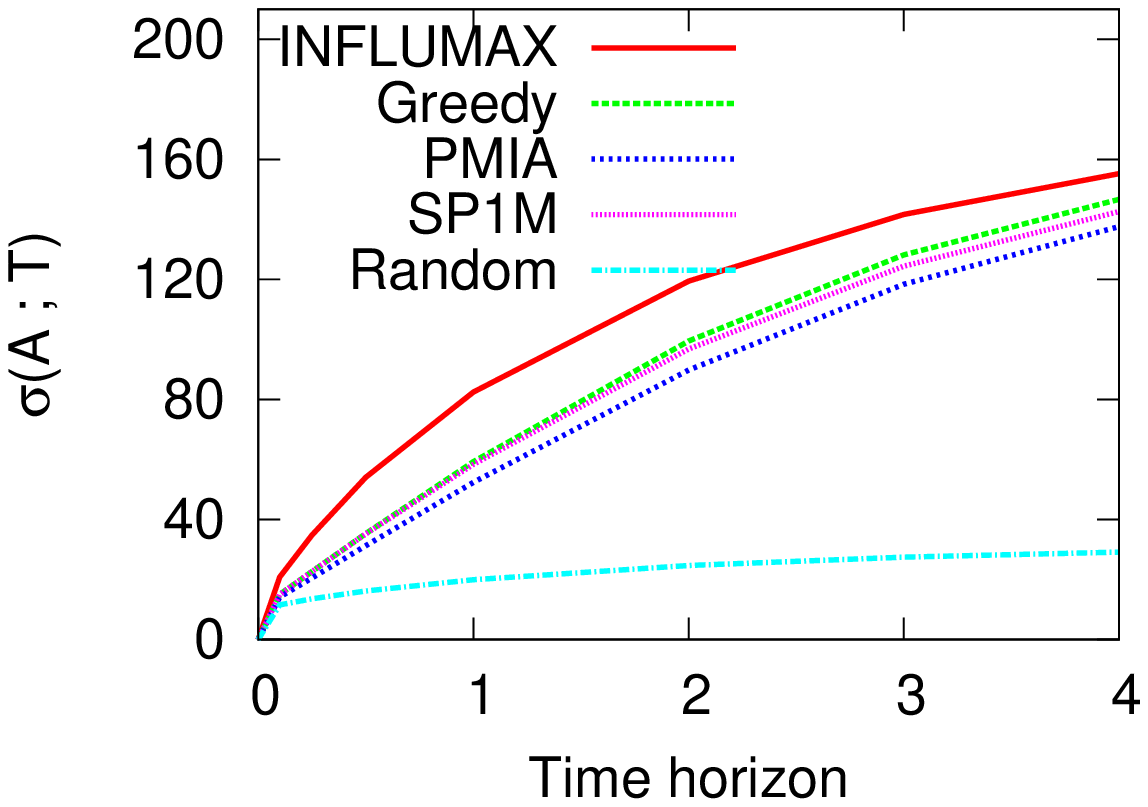} 
 \label{fig:time-horizon}}  \hspace{3mm}
 \subfigure[Running time]{\includegraphics[width=0.2\textwidth]{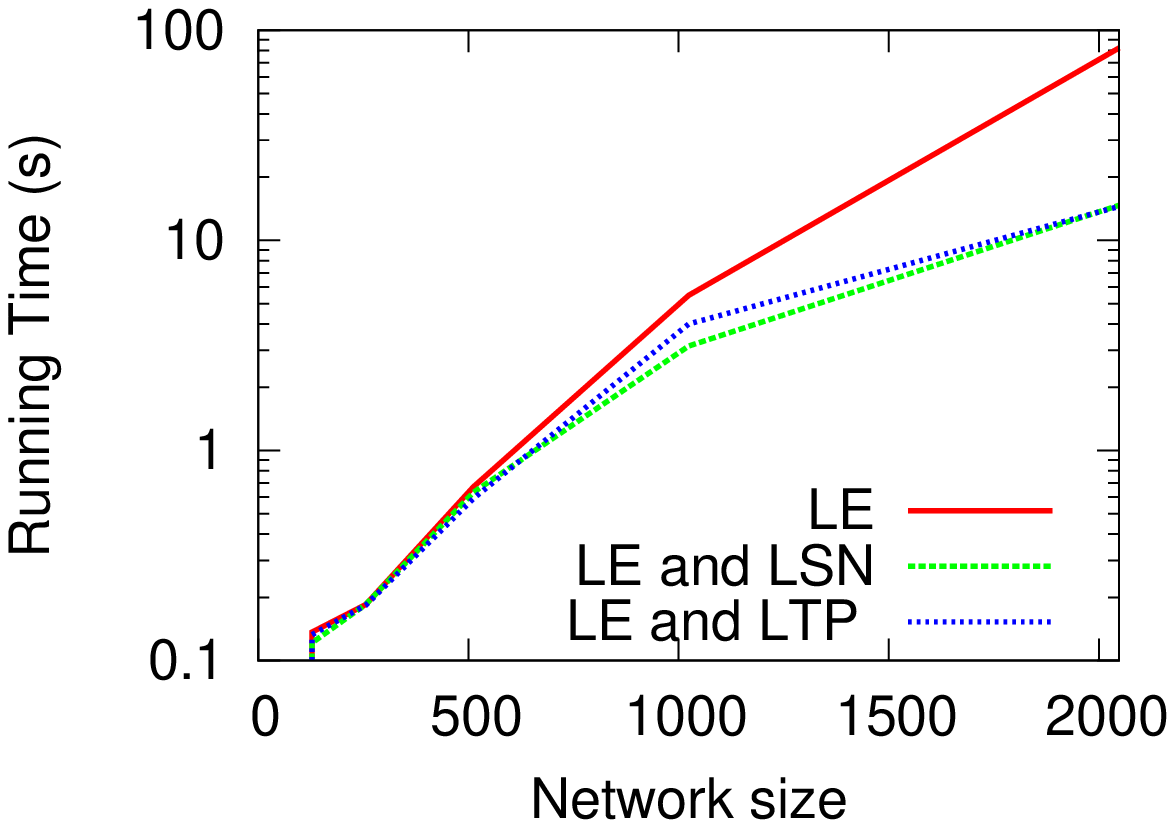} \label{fig:running-time}}
\caption{Panels show (a) influence $\sigma(A; T)$ vs. time horizon and (b) average computation time per source added for \maxinf implemented with (i) lazy 
evaluation (LE), (ii) LE and localized source nodes (LSN, $m = 6$), and (iii) LE and limited transmission paths (LTP, $m = 6$) against number of nodes.}
\end{figure}

Now, we focus on different large synthetic networks. Fi\-gure~\ref{fig:performance} shows the average total number of infected nodes against number of sources that \maxinf achieves in 
comparison with the other methods on a 512 node random Kronecker network, a 1,024 node hierarchical Kronecker network and a 1,024 node Forest Fire (scale free) network. All
three networks have approximately $2$ edges in average per node. We set the time horizon to $T = 1.0$ and the transmission rates are drawn from a uniform distribution $\alpha \sim U(0, 5)$. 
\maxinf typically outperforms other methods by at least 20\% by exploiting the temporal dynamics of the network. We also compare \maxinf with the on-line bound from 
Th.~\ref{th:online-bound}. Fig.~\ref{fig:online-bound} shows the average number of infected nodes against number of sources that \maxinf achieves in comparison with the on-line bound 
for the small core-periphery Kronecker network and the large hierarchical Kronecker network that we used previously. If we pay attention to the value of the bound on the small network for source 
set sizes significantly smaller than the number of nodes in the network, we observe that the bound value on the influence is not as close to the optimal value given by exhaustive search as we could expect. That means that although the bound is not very tight on the large network, we may be actually achieving in practice an almost optimal value on that network too.
\begin{figure}[!t]
\centering
 \subfigure[Hyperlink cascades]{\includegraphics[width=0.2\textwidth]{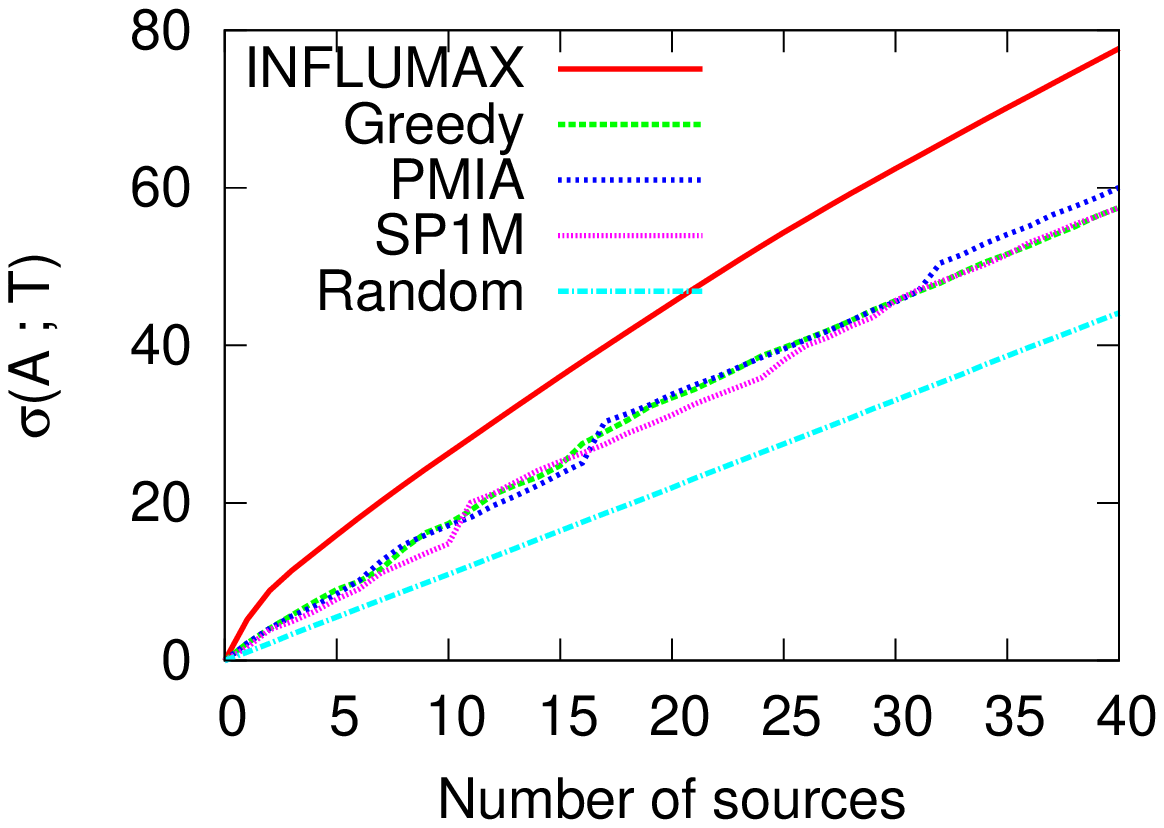}} \hspace{3mm}
 \subfigure[MemeTracker cascades]{\includegraphics[width=0.2\textwidth]{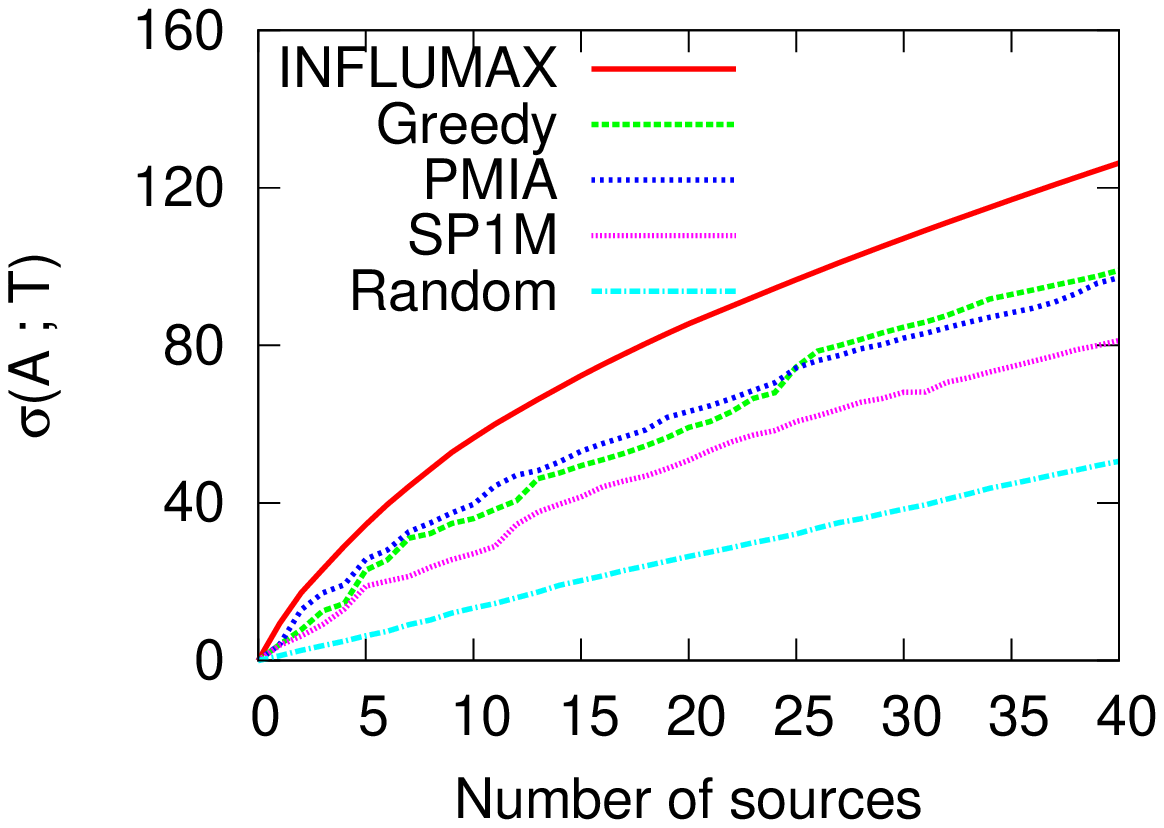}}
	\caption{Influence $\sigma(A ; T)$ for time horizon $T = 1$ against number of sources for 
	(a) a 1,000 node real diffusion network that we infer from hyperlinks cascades and
	(b) a 1,000 node real diffusion network that we infer from MemeTracker cascades.
	The proposed algorithm \maxinf outperforms all other methods by 20-25\%.}
	\label{fig:real}
\end{figure}

\xhdr{Influence vs. time horizon} Intuitively, the smaller the time horizon, the more important the temporal dynamics become when choosing the subset of most influential nodes of a 
given size. Fig.~\ref{fig:time-horizon} shows the average total number of infected nodes against time horizon for a hierarchical Kronecker network with 1,024 nodes and approx.
2 edges per node. We consider a source set of cardinality \mbox{$|A| = 10$} and we draw the transmission rate of each edge from a uniform distribution $\alpha \sim U(0, 5)$. The experimental 
results for all transmission rates configurations confirm the initial intuition, \ie, the difference between \maxinf and other methods is greater for small time horizons. 

\xhdr{Running time} Fig.~\ref{fig:running-time} shows the average computation time per source added of our algorithm implemented (i) with lazy evaluation, (ii) with lazy evaluation
and localized source nodes with $m = 6$ hops and (iii) with lazy evaluation and limited transmission paths with $m = 6$ hops on a single CPU (2.3 Ghz Dual Core with 4 GB RAM). We use hierarchical Kronecker networks with an increasing number of nodes but approximately the same network density since real networks are usually sparse. 
Remarkably, the number of hops that we use in localized source nodes and limited transmission paths result in an approximation error for the influence $\sigma(A; T)$ of at most
10\%, while achieving an speed-up of $\sim$$5$x for the largest network (2,048 nodes).

\subsection{Experiments on real data}
\xhdr{Experimental setup} We used the publicly available MemeTracker dataset, which contains more than $172$ mi\-llion news articles 
and blog posts~\cite{leskovec2009meme}. We trace the information in two different ways and then infer two different diffusion 
networks using \netrate~\cite{manuel11icml}.

First, we find more than 100,000 hyperlink cascades in the MemeTracker dataset. Each hyperlink cascade consists of a collection of time-stamped hyperlinks 
between sites (in blog posts) that refer to closely related pieces of information. From the hyperlink cascade data, we infer an underlying diffusion networks with 
the top (in terms of hyperlinks) 1,000 media sites and blogs.
Second, we apply the MemeTracker methodology~\cite{leskovec2009meme} to find 343 million short textual phrases. We cluster the phrases to aggregate different 
textual variants of the same phrase and consider the 12,000 largest clusters. Each phrase cluster is a MemeTracker cascade. Each cascade consists of a 
collec\-tion of time-stamps when sites (in blog posts) first mentioned any phrase in the cluster. From the MemeTracker cascades, we infer an underlying diffusion 
network with the top (in terms of phrases) 1,000 media sites and blogs.
Then, we sparsify further the networks by kee\-ping the 1,000 fastest edges since it has been shown that in the context of influence maximization, computations on 
sparsified models give up little accuracy, but improves scalability~\cite{bonchi11kdd}.

\xhdr{Solution quality} Fig.~\ref{fig:real} shows the average total number of infected nodes against number of sources that \maxinf achieves in comparison with other 
methods for both real networks, that were inferred from the hyperlink cascade and the MemeTracker cascade datasets, as described above. We set the time horizon to 
$T = 1.0$. Again, \maxinf outperforms all other methods typically by $\sim$$30$\%, by considering the temporal dynamics of the di\-ffusion.
Finally, we also compare \maxinf with the on-line bound from Th.~\ref{th:online-bound} for the real network that we inferred from the hyperlink cascade dataset in Fig.~\ref{fig:bound-real}. 
Similarly to the synthetic networks, the bound is not as tight as expected.

\section{Conclusions}
\label{sec:conclusions}
We have developed a method for influence maximization, \maxinf, that accounts for the temporal dynamics underlying diffusion processes. The 
method allows for va\-ria\-ble transmission (influence) rates between nodes of a network, as found in real-world scenarios.  Perhaps sur\-prising\-ly, 
for the rather general case of continuous temporal dynamics with variable transmission rates, we can eva\-luate the influence of any set of source nodes 
in a network ana\-ly\-ti\-cally using the work of~\citet{kulkarni1986shortest}. In this 
analytical framework, we find the near-optimal set of nodes that maximizes in\-fluence by exploiting the submodularity of our objective function. In addition, 
the reevaluation of influence for changes on the network is straightforward and the algorithm pa\-ra\-lle\-lizes naturally by sink and source nodes.

We evaluated our algorithm on a wide range of synthetic diffusion networks with heterogeneous temporal dynamics which aim to mimic the structure of 
real-world social and information networks. Our algorithm is remarkably stable across different network topologies. It outperforms state of the art methods 
in terms of influence (\ie, average number of infected nodes) for different network topologies, time horizons and source set sizes. \maxinf typically gives an 
influence gain of $\sim$$25$\% and it achieves the greatest improvement for small time horizons; in such scenarios, the temporal dynamics play a dramatic role.
We also evaluated \maxinf on two real diffusion networks that we inferred from the MemeTracker dataset using \netrate. Again, it drastically outperformed 
the state of the art by $\sim$$30$\%.

We believe that \maxinf provides a novel view of the in\-fluence maximization problem by accoun\-ting for the underlying temporal dynamics of diffusion 
networks.

\bibliographystyle{icml2012}
\bibliography{refs}

\end{document}